\documentclass[a4paper,11pt,final]{article}
\usepackage[utf8]{inputenc}
\usepackage[english]{babel}
\usepackage[T1]{fontenc}
\usepackage{lmodern}

\usepackage{amsmath}            
\usepackage{amsfonts}           
\usepackage{amsthm}             
\usepackage{amssymb}             
\usepackage{graphicx}           
\usepackage{natbib}             

\usepackage[nottoc]{tocbibind}  
\usepackage{float}              
\usepackage[colorinlistoftodos,prependcaption,textsize=tiny,obeyFinal]{todonotes}
\usepackage{xargs}              

\newcommandx{\unsure}[2][1=]{\todo[linecolor=red,backgroundcolor=red!25,bordercolor=red,#1]{#2}}
\newcommandx{\change}[2][1=]{\todo[linecolor=orange,backgroundcolor=orange!25,bordercolor=orange,#1]{#2}}
\newcommandx{\info}[2][1=]{\todo[linecolor=green,backgroundcolor=green!25,bordercolor=green,#1]{#2}}
\newcommandx{\improvement}[2][1=]{\todo[linecolor=blue,backgroundcolor=blue!25,bordercolor=blue,#1]{#2}}

\theoremstyle{plain}
\newtheorem{theorem}{Theorem}
\newtheorem{example}[theorem]{Example}

\DeclareMathOperator{\dive}{div}

\newcommand{\x}{(\mathbf{x})}

\newcommand{\sy}{(\mathbf{s}^\mathbf{y})}
\newcommand{\xt}{(\mathbf{x}, t)}
\newcommand{\xty}{(\mathbf{x}^\mathbf{y}, t)}
\newcommand{\sty}{(\mathbf{s}^\mathbf{y}, t)}
\newcommand{\xxt}{(\mathbf{x'}, \mathbf{x}, t)}

\newcommand{\xx}{(\mathbf{x'}, \mathbf{x})}
\newcommand{\xxy}{(\mathbf{x'}^\mathbf{y}, \mathbf{x}^\mathbf{y})}
\newcommand{\ABt}{(A, B, t)}
\newcommand{\ABty}{(A^\mathbf{y}, B^\mathbf{y}, t)}

\newcommand{\ABy}{(A^\mathbf{y}, B^\mathbf{y})}
\newcommand{\snt}{(\mathbf{s}, \mathbf{n}, t)}
\newcommand{\snty}{(\mathbf{s}^\mathbf{y}, \mathbf{n}^\mathbf{y}, t)}
\newcommand{\sn}{(\mathbf{s}, \mathbf{n})}
\newcommand{\sny}{(\mathbf{s}^\mathbf{y}, \mathbf{n}^\mathbf{y})}
\newcommand{\dx}{\mathrm{d}}

\title{The~Peridynamic Stress Tensors\\ and~the~Non-local to~Local Passage}

\author{Petr Pelech\footnote{Mathematical Institute, Charles University, Sokolovská 83, CZ-186 75 Praha 8 (pelech@karlin.mff.cuni.cz )}}

\begin{document}

\maketitle
\begin{abstract}
\noindent
We re-examine the~notion of~stress in~peridynamics.
Based on~the~idea of~traction
we define two new peridynamic stress tensors $\mathbf{P}^{\mathbf{y}}$ and $\mathbf{P}$ 
which stand, respectively, for~analogues of~the~Cauchy and 1st Piola-Kirchhoff stress tensors from classical elasticity.
We show that the~tensor $\mathbf{P}$ differs from the~earlier defined peridynamic stress tensor $\nu$;
though their divergence is equal.
We address the~question of~symmetry of the~tensor $\mathbf{P}^{\mathbf{y}}$
which proves to~be symmetric in~case of~bond-based peridynamics;
as~opposed to~the~inverse Piola transform of~$\nu$ (corresponding to~the~analogue of~Cauchy stress tensor)
which fails to~be symmetric in~general.
We also derive a~general formula of~the~force-flux in~peridynamics
and compute the~limit of~$\mathbf{P}$ for~vanishing non-locality, denoted by~$\mathbf{P}_0$.
We show that this tensor $\mathbf{P}_0$ surprisingly coincides with the~collapsed tensor $\nu_0$,
a~limit of~the~original tensor~$\nu$.
At~the~end, using this flux-formula, we suggest an~explanation
why the~collapsed tensor $\mathbf{P}_0$ (and hence $\nu_0$)
can be indeed identified with the~1st Piola-Kirchhoff stress tensor.

\end{abstract}


\section*{Introduction}
Peridynamics is a~non-local model in~continuum mechanics
introduced in~\cite{Silling2000} and elaborated later in~\cite{Silling2007}
(see also \cite{SillingLehoucq2010} or \cite{EmmrichLehoucq2012} for survey of most important results).
The non-locality is reflected in~the~fact
that points at~a~finite distance exert a~force upon each other.
This force interaction is~described by~a~\emph{pairwise force function}
$\mathbf{f}:\Omega\times\Omega\times[0,T] \rightarrow \mathbb{R}^3$,
where $\Omega \subset \mathbb{R}^3$ denotes the~body in~the~reference configuration
and~$[0,T]$, with $T>0$, is the~time interval of~interest.
If, however,
the points are in~the~reference configuration more distant than a~characteristic length called \emph{horizon},
it is customary to assume that they do not interact.
This is~stated in~the~assumption that
\begin{align} \label{eq:PFFHorizon}
  \mathbf{f}\xxt = \mathbf{0}
  \quad \text{whenever} \quad |\mathbf{x}' - \mathbf{x}| \geq \delta,
\end{align}
where $\delta > 0$ denotes the~horizon length (or~just the~horizon for~simplicity).
The~meaning of~horizon is that it represents an~internal material length scale
(see e.g. \cite{BobaruHu2012}, \cite{SillingLehoucq2008} and \cite{SillingLehoucq2010}
for further explanation and examples).
The~force $F(A, B, t)$ which one part of~the~body $A \subset \Omega$
exerts on~another part $B \subset \Omega$ at~time $t$
is a~summation of~all point interactions
\begin{align} \label{eq:PFFunction}
  F\ABt
  = \int_B \int_A
      \mathbf{f}\xxt \,
    \dx \mathbf{x}' \dx \mathbf{x}.
\end{align}
Hence even disjoint parts $A, B \subset \Omega$ may interact with each other.
From this formula it is also obvious that the dimension of~$\mathbf{f}$ is force per volume squared.
The~equation of~motion in~peridynamics then takes the~form
\begin{align} \label{eq:ForceBal}
  \rho_0\x \ddot{\mathbf{y}}\xt
  = \int_\Omega \mathbf{f}\xxt \,
    \dx \mathbf{x}'
    +
    \mathbf{b}\xt,
  \quad \text{for all } \xt \in \Omega\times[0, T],
\end{align}
where $\ddot{\mathbf{y}}$ is~the~second time derivative of~the~deformation
$\mathbf{y}:\Omega\times[0,T] \rightarrow \mathbb{R}^3$,
$\rho_0:\Omega \rightarrow (0, +\infty)$ the~density in~the~reference configuration,
and~$\mathbf{b}:\Omega\times[0,T] \rightarrow \mathbb{R}^3$ the~density of~external forces
with respect to~the~volume in~the~reference configuration.
The~specific form of~$\mathbf{f}$ is matter of~a~constitutive theory
and it usually involves deformation in~a~non-local way.

This is in~contrast to~standard local theories of~simple materials
(cf. \cite{GurtinAnand2010} or \cite{Ciarlet1988})
where two \emph{adjacent} parts of~the~deformed body interacts through a~common surface.
The~interaction is~described by~the~Cauchy stress vector
$\mathbf{t}^\mathbf{y}:\mathbf{y}(\Omega)\times\mathbb{S}^2\times[0, T] \rightarrow \mathbb{R}^3$
which depends on~the~position in~the~deformed configuration,
the~surface normal vector at~that point ($\mathbb{S}^2$ denotes the~unit sphere in~$\mathbb{R}^3$),
and~time.
This vector represents the~surface density of~that force interaction.
Hence the~force between two~adjacent spatial regions
$A^\mathbf{y} \subset \mathbf{y}(\Omega)$
and~$B^\mathbf{y} \subset \mathbf{y}(\Omega)$ at~time $t$
is~expressed by~the~surface integral
\begin{align} \label{eq:CVector}
  F\ABty
  = \int_{\partial A^\mathbf{y} \cap \partial B^\mathbf{y}}
      \mathbf{t}^\mathbf{y}\snty \,
    \dx S(\mathbf{s}^\mathbf{y}),
\end{align}
where $\mathbf{n}^\mathbf{y}$ denotes the~outer normal at~the~point $\mathbf{s}^\mathbf{y}$.
By~the~Cauchy theorem there exist the~Cauchy stress tensor
$\mathbf{T}^\mathbf{y}:\mathbf{y}(\Omega)\times[0,T] \rightarrow \mathbb{R}^{3\times 3}$
such that
\begin{align} \label{eq:CTensor}
  \mathbf{t}^\mathbf{y}\snty
  = \mathbf{T}^\mathbf{y}\sty \mathbf{n}^\mathbf{y}
\end{align}
for~all $\snty \in \mathbf{y}(\Omega)\times\mathbb{S}^2\times[0, T]$,
i.e. the~dependence on~the~unit normal is linear.
The~Gauss theorem then implies that
\begin{align} \label{eq:CGauss}
  \int_{A^\mathbf{y}}
    \dive^\mathbf{y} \mathbf{T}^\mathbf{y}\xty \,
  \dx \mathbf{x}^\mathbf{y}
  = \int_{\partial A^\mathbf{y}}
      \mathbf{T}^\mathbf{y}\sty \mathbf{n}^\mathbf{y} \,
    \dx S(\mathbf{s}^\mathbf{y})
\end{align}
for~any $t \in [0, T]$ and~$A^\mathbf{y} \subset \mathbf{y}(\Omega)$ smooth enough,
where $\dive^\mathbf{y}$ means the~divergence with~respect to~the~spatial variables
in~the~deformed configuration.
Therefore, using \eqref{eq:CTensor} and~\eqref{eq:CVector},
the~divergence of~the~Cauchy tensor expresses the~volume density of~internal forces
with respect to~the~volume in~the~deformed configuration.
The~first Piola~Kirchhoff stress tensor
$\mathbf{T}:\Omega\times[0,T] \rightarrow \mathbb{R}^{3\times 3}$ is~defined
as~the~Piola transform of~the~Cauchy tensor
\begin{align} \label{eq:PiolaTrafo}
  \mathbf{T}\xt
  = (\det \nabla \mathbf{y}\xt)
    \mathbf{T}^\mathbf{y}\xty
    (\nabla \mathbf{y}\xt)^{-\top},
    \quad
    \mathbf{x}^\mathbf{y} = \mathbf{y}(\mathbf{x}).
\end{align}
Thanks to~the~properties of~this transform,
the~following equality holds
\begin{align*}
  \int_{A}
    \dive \mathbf{T}\xt \,
  \dx \mathbf{x}
  = \int_{A^\mathbf{y}}
      \dive^\mathbf{y} \mathbf{T}^\mathbf{y}\xty \,
    \dx \mathbf{x}^\mathbf{y}
\end{align*}
for~any~$t \in [0, T]$, where $A \subset \Omega$, $A^\mathbf{y} = \mathbf{y}(A)$,
and~'$\dive^\mathbf{y}$' stands for~spatial divergence in~the~reference configuration.
The~divergence of~the~first Piola-Kirchhoff stress tensor
then expresses the~density of~internal forces
with respect to~the~volume in~the~reference configuration.
Hence the~equation of~motion in~the~reference configuration takes the~form
\begin{align*}
  \rho_0\x \ddot{\mathbf{y}}\xt
  = \dive \mathbf{T}\xt
    +
    \mathbf{b}\xt,
  \quad \text{for all } \xt \in \Omega\times[0, T].
\end{align*}

The~divergence of~any stress tensor
provides knowledge only of~the~total force flux through closed surfaces
which is, however, not sufficient for~building the~whole theory.
For~the~formulation of~the~balance of~angular momentum the~whole tensor is needed.
It can be shown that this balance is equivalent to~the~symmetry of~the~Cauchy stress tensor,
i.e.
\begin{align*}
  \mathbf{T}^\mathbf{y} = {\mathbf{T}^\mathbf{y}}^{\top}.
\end{align*}
A~corresponding condition for~the~first Piola-Kirchhoff
tensor can be derived from \eqref{eq:PiolaTrafo}.


The~question whether such different concepts of~interaction can be~related to~each other
was addressed already in~the~pioneering work \cite{Silling2000}.
Here the~\emph{areal force density} at~a~point $\mathbf{s} \in \Omega$ and~time $t \in [0, T]$
in~the~direction of~unit vector $\mathbf{n}$ in~the~reference configuration
is defined as
\begin{align*}
  \tau\snt
  = \int_{\mathcal{L}\sn} \int_{\Omega^+\sn}
      \mathbf{f}(\mathbf{x}', \hat{\mathbf{x}}, t) \,
    \dx \mathbf{x}' \dx l(\hat{\mathbf{x}}),
\end{align*}
where
\begin{align*}
  \Omega^+\sn &= \{ \mathbf{x}' \in \Omega: (\mathbf{x}' - \mathbf{s}) \cdot \mathbf{n} \geq 0 \},
  \\
  \mathcal{L}\sn &= \{ \hat{\mathbf{x}} \in \Omega: \exists s \geq 0 \text{ s.t. }\hat{\mathbf{x}} = \mathbf{s} - s\mathbf{n} \},
\end{align*}
and~'$\dx l$' represents a~length element.
As~was already mentioned in~\cite{Silling2000},
this definition of~$\mathbf{\tau}$ is most useful
in~the~case of~a~homogeneous deformation
(i.e. a~deformation whose gradient is a~constant matrix).
Moreover, its linear dependence on~$\mathbf{n}$ is not obvious,
and~it provides no~explicit formula for~the~first Piola-Kirchhoff tensor
in~terms of~the~pairwise force function $\mathbf{f}$.
%
%
This~issue was somehow overcome later in~\cite{SillingFlux2008}
where the~\emph{peridynamic stress tensor} $\mathbf{\nu}$ was defined as
\begin{align} \label{eq:nu}
  \mathbf{\nu} \xt
  =
  \frac{1}{2}
    \int\limits_{\mathbb{S}^{2}} \int\limits_0^{+\infty} \int\limits_0^{+\infty}
      (\alpha + \beta)^2 \mathbf{f}(\mathbf{x} + \alpha\mathbf{m}, \mathbf{x} - \beta\mathbf{m}, t) \otimes \mathbf{m} \,
    \dx \alpha \dx \beta \dx S(\mathbf{m}),
\end{align}
though no~connection to~areal force density was provided.
%
%
Note that through the~function $\mathbf{f}$ the~dependence of~$\nu$ on~the~deformation is non-local.
The~spatial divergence of~this tensor is equal to~the~density of~internal
forces in~the~reference configuration, i.e.
\begin{align*}
  \dive \mathbf{\nu} \xt
  = \int_\Omega \mathbf{f}\xxt \,
    \dx \mathbf{x}',
  \quad
  \xt \in \Omega\times [0, T],
\end{align*}
and~so~the~equation of~motion in~peridynamics can be~rewritten to~the~form
\begin{align*}
  \rho_0\x \ddot{\mathbf{y}}\xt
  = \dive \mathbf{\nu}\xt
    +
    \mathbf{b}\xt,
  \quad \text{for all } \xt \in \Omega\times[0, T]
\end{align*}
which is~formally similar to~the~equation of~motion in~conventional theory.
This fact is~subsequently used in~\cite{SillingLehoucq2008}
where the~convergence of~peridynamics to~the conventional theory
is~investigated for the~horizon tending to~zero.
After performing a~scaling,
the measure of~the~non-locality is represented by~a~dimensionless paremeter $s \rightarrow 0$.
It~is shown that for~a~sufficiently smooth fixed deformation and constitutive relation
\begin{align*}
  \lim_{s \rightarrow 0} \mathbf{\nu} \xt = \mathbf{\nu}_0 \xt,
  \quad \xt \in \Omega \times [0,T].
\end{align*}
The~tensor $\mathbf{\nu}_0\xt$
is called the~\emph{collapsed peridynamic tensor}
and~it depends on~the~deformation only through the~deformation gradient at~a~point.

If one is interested in~the~convergence of~the equation of~motion alone,
then the~description provided by~$\dive \nu \xt$ and~$\dive \nu_0 \xt$ is~sufficient.
Nevertheless, for~identifying the~limiting model in~standard elasticity
the~knowledge of~the~whole limiting first Piola-Kirchhoff stress tensor $\mathbf{T}_0$ is~necessary.
Since the~$\dive \nu$ expresses the~volume density of~internal forces,
its limit $\dive \nu_0$ express the~density of~internal forces in~the~limiting model
and therefore it~must hold
\begin{align*}
  \dive \mathbf{T}_0 = \dive \nu_0.
\end{align*}
This requirement, however, determines the~tensor $\mathbf{T}_0$ only up~to~an~additive so\-le\-noi\-dal tensor field.
Hence a~closer connection between these two stress tensors is~needed.

The~same complication is~related to~the peridynamic tensor $\nu$
which is~suggested in~\cite{SillingFlux2008} as~an~analogue
of~the~first Piola-Kirchhoff stress tensor $\mathbf{T}$.
This suggestion is~also based only on~the~divergence of~$\nu$
which is, for~reason mentioned above, not sufficient for~the~identification of~the~whole tensor.
In~fact, any such an~identification of~$\nu$ and~$\mathbf{T}$ is~problematic.
One reason is connected to~the~balance of~angular momentum.
If~one considers a~sufficiently regular deformation,
then the~inverse Piola transform of~$\nu$ (which would be~an~analogue of~the~Cauchy tensor)
is \emph{not}~symmetric in~general.
Even though the~balance of~angular momentum in~peridynamics is satisfied
and only materials with no internal structure
(also called non-polar materials in~\cite{SillingLehoucq2010}) are considered;
therefore the~analogy between the~peridynamic tensor
and~the~first Piola-Kirchhoff tensor remains rather unclear.
The~aim of~this paper is to~derive a~notion of~stress in~peridynamics in a~consistent way.
In~particular this means to~define another tensor $\mathbf{P}$
whose divergence would also satisfy
\begin{align*}
  \dive \mathbf{P} \xt
  = \int_\Omega \mathbf{f}\xxt \,
    \dx \mathbf{x}',
\end{align*}
i.e. it~would express the~volume density of~internal forces.
At~the~same time either its inverse Piola transform
would be~symmetric,
or some explanation of~its asymmetry would be~provided.
Last but not least, it would be~directly related to~the~Cauchy~stress \emph{vector}.
Only such a~direct relation would make it~a~real analogue
of~the~first Piola-Kirchhoff stress tensor $\mathbf{T}$,
since from~the~Cauchy stress vector any~other stress measures are~derived.
In~this way the~tensor $\mathbf{P}$ would also generalize the~areal force density $\tau\snt$
which is useful only for~homogeneous deformations
and~whose linear dependence on~the~vector $\mathbf{n}$ is unclear.

In~addition, the~tensor $\mathbf{P}$ can be used for~investigating the~limiting
behavior of~peridynamics for~vanishing non-locality in~the~same spirit
as~the~tensor $\nu$ was used.
Thanks to~its relation to~the~Cauchy stress vector,
its limiting counterpart $\mathbf{P_0}$ may~be~identified with~the~first Piola-Kirchhoff stress tensor,
and~hence it~may provide a~full description of~the~limiting model in~classical theory.

Although discontinuous deformations are~possible in~peridynamics,
we restrict ourselves to~sufficiently smooth deformations
for which all objects from~both theories are well defined.
At the~same time, the~choice of~optimal function spaces and~control volumes
is~left for~further investigation.
The~structure of~the~article is~as~follows.
In Section \ref{sec:derivation} we motivate our definition of peridynamic tensors $\mathbf{P}^{\mathbf{y}}$ and $\mathbf{P}$.
In~Section \ref{sec:properties} we show some of their properties, compare them with the peridynamic tensor $\nu$
and compute the limit of~$\mathbf{P}$ for~horizon tending to~zero.
In~the last section we derive a~general formula of~force flux in~peridynamics
which is then used for a~comparison between Cauchy stress tensor $\mathbf{T}^{\mathbf{y}}$
and peridynamic tensor $\mathbf{P}^{\mathbf{y}}$.

\section{Derivation of~the~peridynamic stress tensors} \label{sec:derivation}
We shall define two peridynamic tensors, denoted by~$\mathbf{P}^\mathbf{y}$ and~$\mathbf{P}$,
which are an~analogues of~the~Cau\-chy and~the~first Piola-Kirchhoff stress tensors respectively.
The~definition of~the~first tensor $\mathbf{P}^\mathbf{y}$
is based on~a~heuristic derivation of~an~analogue of~the~Cauchy stress vector $\mathbf{t}^\mathbf{y}$
in~terms of~the~pairwise force function $\mathbf{f}$.
The~desired peridynamic tensor $\mathbf{P}$ is then defined as~the~Piola transform
of~$\mathbf{P}^\mathbf{y}$.

In~order to~compute the~vector $\mathbf{t}^\mathbf{y}$
we~divide the~deformed body by~a~plane into~two pieces.
The~mutual force interaction between these parts
is given by a~double volume integral in~\eqref{eq:PFFunction}.
Using a~suitable substitution we rewrite it as~a~surface integral over the~dividing plane.
The~vector $\mathbf{t}^\mathbf{y}$ is set to~be equal to~the~corresponding surface density.
This~way of~deriving the~formula for~$\mathbf{t}^\mathbf{y}$ seems to~be similar
to~the~one used in~\cite{Silling2000} for~obtaining the~expression for~the~areal force density $\tau$.
The~difference is that here all the~computations are~performed in~the deformed configuration
and~the~integration over the dividing surface is~done in~a~different manner.
The~former makes the~derivation meaningful even for~non-homogeneous deformations,
the~latter results in~a~formula for~$\mathbf{t}^\mathbf{y}$
from which the~form of~the~stress tensor $\mathbf{P}^\mathbf{y}$ is~explicitly visible.

Since we want to~compute in~the~deformed configuration,
we define the~vector field $\mathbf{f}^\mathbf{y}$
as~the~corresponding density with respect to~the~volume in~the~deformed configuration
i.e.
\begin{align} \label{eq:f=fy}
  \mathbf{f}\xx
  =
  \mathbf{f}^\mathbf{y}\xxy
  (\det \nabla \mathbf{y}(\mathbf{x'})) (\det \nabla \mathbf{y}\x),
  \quad
  \mathbf{x'}^\mathbf{y} = \mathbf{y}(\mathbf{x'}),\,
  \mathbf{x}^\mathbf{y} = \mathbf{y}(\mathbf{x})
\end{align}
and
\begin{align*}
  \int_B \int_A
      \mathbf{f}\xx \,
    \dx \mathbf{x}' \dx \mathbf{x}
  = \int_{B^\mathbf{y}} \int_{A^\mathbf{y}}
      \mathbf{f}^\mathbf{y}\xxy \,
    \dx \mathbf{x'}^\mathbf{y} \dx \mathbf{x}^\mathbf{y}.
\end{align*}
From now on, for~the sake of~brevity, the~time argument will be suppressed.
For~later simplicity, we set
\begin{align}
  \mathbf{f}\xx
  &= \mathbf{0},
  &
  \text{whenever }& \mathbf{x}' \notin \Omega \text{ or } \mathbf{x} \notin \Omega, \\ \label{eq:fy=0}
  \mathbf{f}^\mathbf{y}\xxy
  &= \mathbf{0},
  &
  \text{whenever }& \mathbf{x'}^\mathbf{y} \notin \mathbf{y}(\Omega)
  \text{ or } \mathbf{x}^\mathbf{y} \notin \mathbf{y}(\Omega).
\end{align}
Consider now an~arbitrary plane $\mathcal{P}^\mathbf{y} \subset \mathbb{R}^3$ in~the~deformed configuration
which has a~normal vector $\mathbf{n}^\mathbf{y} \in \mathbb{S}^2$
and divides the~body $\mathbf{y}(\Omega)$ into two pieces.
Without loss of~generality we may choose a~Cartesian coordinate system
$\mathbf{x}^\mathbf{y} = (x^y_1, x^y_2, x^y_3)$
such that $\mathbf{n}^\mathbf{y} = (1, 0, 0)$ and~$(0, 0, 0) \in \mathcal{P}^\mathbf{y}$.
The~two parts of~the~deformed body are then given by
\begin{align} \label{eq:yOmega+}
  \mathbf{y}(\Omega)_+ = \{ \mathbf{x}^\mathbf{y} \in \mathbf{y}(\Omega): \mathbf{n}^\mathbf{y} \cdot \mathbf{x}^\mathbf{y} = x^y_1 > 0\}, \\
  \label{eq:yOmega-}
  \mathbf{y}(\Omega)_- = \{ \mathbf{x}^\mathbf{y} \in \mathbf{y}(\Omega): \mathbf{n}^\mathbf{y} \cdot \mathbf{x}^\mathbf{y} = x^y_1 < 0\}.
\end{align}
The~force which one part exerts on~the~other is~then expressed as
\begin{align*}
  F(\mathbf{y}(\Omega)_+, \mathbf{y}(\Omega)_-)
  = \int_{\mathbf{y}(\Omega)_-} \int_{\mathbf{y}(\Omega)_+}
      \mathbf{f}^\mathbf{y}\xxy \,
    \dx \mathbf{x'}^\mathbf{y} \dx \mathbf{x}^\mathbf{y}.
\end{align*}
The~line segment $[\mathbf{x'}^\mathbf{y}, \mathbf{x}^\mathbf{y}]$ given by~the~couple of~interacting points
intersects the~dividing plane $\mathcal{P}^\mathbf{y}$ at~a~unique point $\mathbf{s}^\mathbf{y}$.
The~line segment $[\mathbf{s}^\mathbf{y}, \mathbf{x'}^\mathbf{y}]$ has the~length $\alpha$
and~points in~the~outer direction $\mathbf{m}$, i.e. $\mathbf{m} \cdot \mathbf{n}^\mathbf{y} > 0$
(see Fig. \ref{fig:FluxPlane}).
The~line segment $[\mathbf{s}^\mathbf{y}, \mathbf{x}^\mathbf{y}]$ has the~length $\beta$
and~points in~the~opposite direction.
This gives rise to~a~substitution
\footnote{This calculation presented here resembles to~the~one used already by~Cauchy
(see \cite{Cauchy1828} or~\cite{Love1892}).
Since we are nevertheless interested in~large deformation in~general and~the~horizon cannot be~considered infinitesimal,
we have to~proceed differently.
}
\begin{align} \label{eq:CauchySubst}
  \mathbf{x'}^\mathbf{y} = \mathbf{s}^\mathbf{y} + \alpha \mathbf{m},
  \quad  
  \mathbf{x}^\mathbf{y} = \mathbf{s}^\mathbf{y} - \beta \mathbf{m},
\end{align}
by which the~integration over all interacting couples
$[\mathbf{x'}^\mathbf{y}, \mathbf{x}^\mathbf{y}] \in \mathbf{y}(\Omega)_+ \times \mathbf{y}(\Omega)_-$
can be~rewritten as~a~surface integral over the~contact plane $\mathcal{P}^\mathbf{y}$ of~a~corresponding surface density.
This surface density is~the~sought vector $\mathbf{t}^\mathbf{y}$.

\begin{figure}
  \centering
  \includegraphics[width=0.5\textwidth]{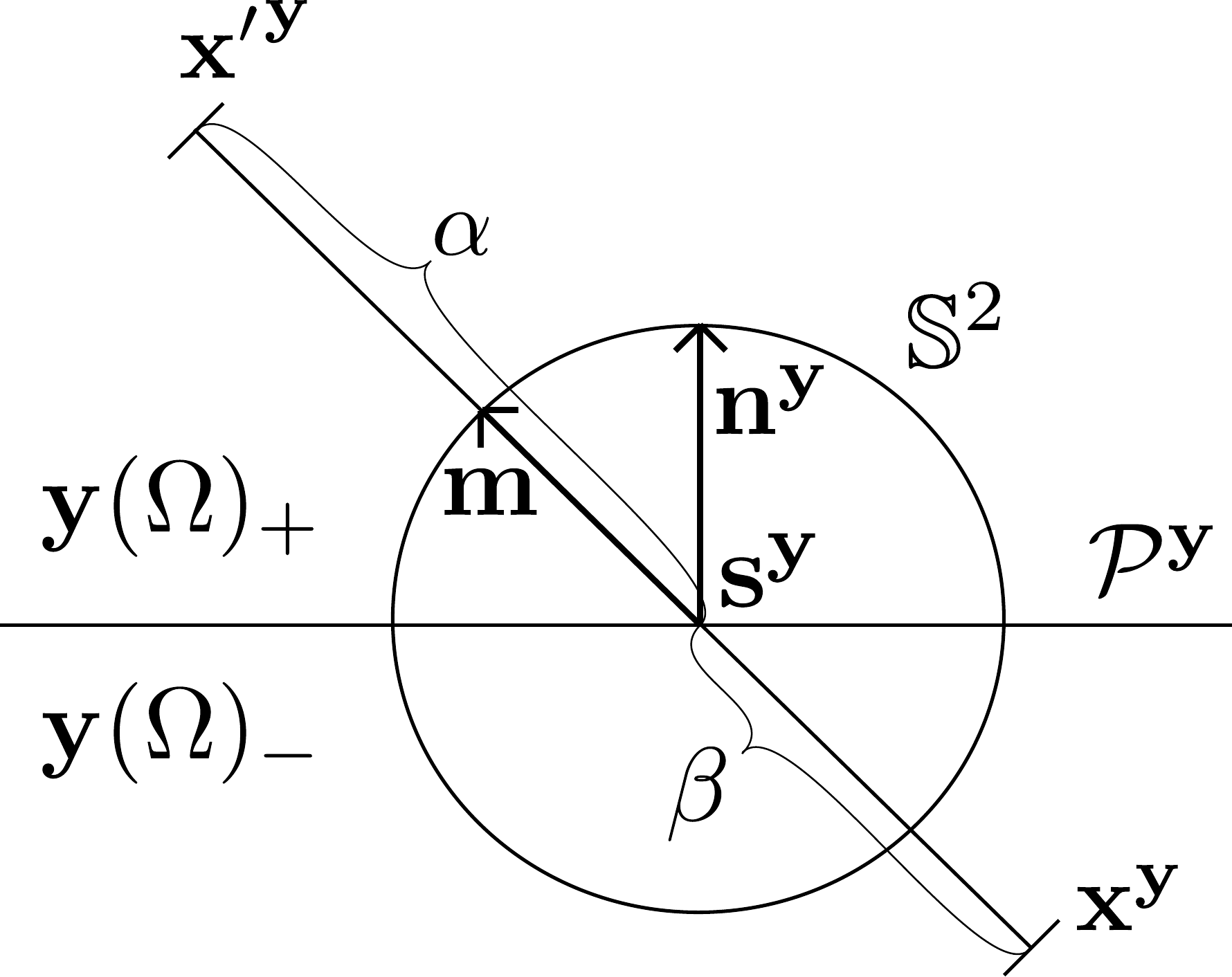}
  \caption{The~force flux through the~plane $\mathcal{P}^\mathbf{y}$}
  \label{fig:FluxPlane}
\end{figure}

In~order to~perform the~substitution properly
we rewrite the~double integral using the~coordinates
(note the~different limits of~integration for~the~coordinates $\dx x'^y_1$ and $\dx x^y_1$
which are due to~\eqref{eq:yOmega+} and \eqref{eq:yOmega-} respectively)
\begin{multline*}
  F(\mathbf{y}(\Omega)_+, \mathbf{y}(\Omega)_-) \\
  =
  \int\limits_{-\infty}^{+\infty} \int\limits_{-\infty}^{+\infty} \int\limits_{-\infty}^{0}
  \int\limits_{-\infty}^{+\infty} \int\limits_{-\infty}^{+\infty} \int\limits_{0}^{+\infty}
    \mathbf{f}^\mathbf{y}\xxy \,
  \dx x'^y_1 \, \dx x'^y_2 \, \dx x'^y_3 \,
  \dx x^y_1 \, \dx x^y_2 \, \dx x^y_3 ,
\end{multline*}
where we use the~abbreviation
\begin{align*}
  \mathbf{x'}^\mathbf{y} = (x'^y_1, x'^y_2, x'^y_3),
  \quad
  \mathbf{x}^\mathbf{y} = (x^y_1, x^y_2, x^y_3).
\end{align*}
Note that the~equality \eqref{eq:fy=0} allows us to~integrate up to~infinity.
The~substitution in~coordinates then takes the~form
\begin{align*}
  x'^y_1 &= \alpha \cos \theta, &
  x'^y_2 &= s^y_2 + \alpha \sin \theta \cos \varphi, &
  x'^y_3 &= s^y_3 + \alpha \sin \theta \sin \varphi, \\
  x^y_1 &= -\beta \cos \theta, &
  x^y_2 &= s^y_2 - \beta \sin \theta \cos \varphi, &
  x^y_3 &= s^y_3 - \beta \sin \theta \sin \varphi,
\end{align*}
where
\begin{align*}  \theta \in (0, \frac{\pi}{2}), \quad
  \varphi \in (0, 2\pi), \quad
  \alpha, \beta \in (0, +\infty), \quad
  s^y_2, s^y_3, \in (-\infty, +\infty).
\end{align*}
The~corresponding Jacobian is
\begin{multline} \label{eq:Jacobian}
  J =\left|
        \begin{array}{ c c c c c c }
          -\alpha \sin\theta            & 0                              & \cos\theta                & 0                         & 0 & 0 \\
          \alpha \cos\theta \cos\varphi & -\alpha \sin\theta \sin\varphi & \sin \theta \cos \varphi  & 0                         & 1 & 0 \\
          \alpha \cos\theta \sin\varphi &  \alpha \sin\theta \cos\varphi & \sin \theta \sin \varphi  & 0                         & 0 & 1 \\
          \beta \sin\theta              & 0                              & 0                         & -\cos\theta               & 0 & 0 \\
          -\beta \cos\theta \cos\varphi & \beta \sin\theta \sin\varphi   & 0                         & -\sin \theta \cos \varphi & 1 & 0 \\
          -\beta \cos\theta \sin\varphi & -\beta \sin\theta \cos\varphi  & 0                         & -\sin \theta \sin \varphi & 0 & 1 \\
        \end{array}
      \right| \\
  = -(\alpha + \beta)^2 \cos\theta \sin\theta.
\end{multline}
The~force is then equal to
\begin{align*}
  &F(\mathbf{y}(\Omega)_+, \mathbf{y}(\Omega)_-) \\
  &\quad =
  \int\limits_{-\infty}^{+\infty} \int\limits_{-\infty}^{+\infty}
  \int\limits_{0}^{+\infty} \int\limits_{0}^{+\infty}
  \int\limits_{0}^{2\pi} \int\limits_{0}^{\frac{\pi}{2}}
    (\alpha + \beta)^2 \,
    \mathbf{f}^\mathbf{y}(\mathbf{s}^\mathbf{y} + \alpha \mathbf{m},
                          \mathbf{s}^\mathbf{y} - \beta \mathbf{m})
    \cos\theta \sin\theta \\
  & \quad \quad \quad \quad \quad \quad \quad \quad \quad \quad \quad \quad \quad \quad \quad \quad \quad \quad \quad \quad \quad \quad \quad \quad
  \dx \theta \, \dx \varphi \,
  \dx \alpha \, \dx \beta \,
  \dx s^y_2 \, \dx s^y_3,
\end{align*}
where
\begin{align*}
  \mathbf{s}^\mathbf{y} = (0, s^y_2, s^y_3),
  \quad
  \mathbf{m} = (\cos\theta, \sin\theta \cos\varphi, \sin\theta \sin\varphi).
\end{align*}
Since $\mathbf{m} \cdot \mathbf{n}^\mathbf{y} = \cos \theta$ and~$\sin\theta\, \dx\theta \, \dx\varphi$
is~the~differential solid angle,
the~integral can be~expressed in~a~coordinate-less form
\begin{multline*}
  F(\mathbf{y}(\Omega)_+, \mathbf{y}(\Omega)_-) \\
  =
  \int\limits_{\mathcal{P}^\mathbf{y}}
  \int\limits_{ \mathbb{S}^2_+ }
  \int\limits_{0}^{+\infty} \int\limits_{0}^{+\infty}
    (\alpha + \beta)^2 \,
    \mathbf{f}^\mathbf{y}(\mathbf{s}^\mathbf{y} + \alpha \mathbf{m},
                          \mathbf{s}^\mathbf{y} - \beta \mathbf{m})
    (\mathbf{m} \cdot \mathbf{n}^\mathbf{y}) \,
  \dx \alpha \, \dx \beta \, \dx S(\mathbf{m}) \, \dx S(\mathbf{s}^\mathbf{y}),
\end{multline*}
where
\begin{align*}
  \mathbb{S}^2_+ := \{ \mathbf{m} \in \mathbb{S}^{2}: \mathbf{m}\cdot\mathbf{n}^\mathbf{y} > 0 \}.
\end{align*}
It should be noted that the~integration is not done over the~common boundary
of~the~parts $\mathcal{P}^\mathbf{y} \cap \mathbf{y}(\Omega)$,
since some line segments connecting the~pairs of~interacting points
may intersect the~plane $\mathcal{P}^\mathbf{y}$ outside the~deformed body $\mathbf{y}(\Omega)$.
Yet we set
\begin{multline} \label{eq:tyHemi}
  \mathbf{t}^\mathbf{y}\sny \\
  :=
  \int\limits_{ \mathbb{S}^2_+ }
  \int\limits_{0}^{+\infty} \int\limits_{0}^{+\infty}
    (\alpha + \beta)^2 \,
    \mathbf{f}^\mathbf{y}(\mathbf{s}^\mathbf{y} + \alpha \mathbf{m},
                          \mathbf{s}^\mathbf{y} - \beta \mathbf{m})
    (\mathbf{m} \cdot \mathbf{n}^\mathbf{y}) \,
  \dx \alpha \, \dx \beta \, \dx S(\mathbf{m}),
\end{multline}
but we can now already foresee some aspects of~the~nature of~the~non-local peridynamic interaction
which are treated in~greater detail in~the~Section~\ref{sec:flux}.
Thanks to~the~skew-symmetry of~$\mathbf{f}^\mathbf{y}$ in~its arguments
and to~the~symmetry of~the~integrand in~$\alpha \in (0, +\infty)$ and~$\beta \in (0, +\infty)$,
the~integration over the~opposite hemisphere $\mathbb{S}^2_-$
yields the~same value and~hence
\begin{multline} \label{eq:tyWhole}
  \mathbf{t}^\mathbf{y}\sny \\
  =
  \frac{1}{2}
  \int\limits_{ \mathbb{S}^2 }
  \int\limits_{0}^{+\infty} \int\limits_{0}^{+\infty}
    (\alpha + \beta)^2 \,
    \mathbf{f}^\mathbf{y}(\mathbf{s}^\mathbf{y} + \alpha \mathbf{m},
                          \mathbf{s}^\mathbf{y} - \beta \mathbf{m})
    (\mathbf{m} \cdot \mathbf{n}^\mathbf{y}) \,
  \dx \alpha \, \dx \beta \, \dx S(\mathbf{m}).
\end{multline}
Based on~this result we define the~peridynamic stress tensor $\mathbf{P}^\mathbf{y}$
in~the~following way
\begin{align} \label{eq:Py}
  \mathbf{P}^\mathbf{y}(\mathbf{x}^\mathbf{y})
  =
  \frac{1}{2}
  \int\limits_{ \mathbb{S}^2 }
  \int\limits_{0}^{+\infty} \int\limits_{0}^{+\infty}
    (\alpha + \beta)^2 \,
    \mathbf{f}^\mathbf{y}(\mathbf{x}^\mathbf{y} + \alpha \mathbf{m},
                          \mathbf{x}^\mathbf{y} - \beta \mathbf{m})
    \otimes \mathbf{m} \,
  \dx \alpha \, \dx \beta \, \dx S(\mathbf{m}).
\end{align}
The~peridynamic stress tensor $\mathbf{P}$ is defined via Piola transform as
\begin{align*}
  \mathbf{P} (\mathbf{x})
  :=
  (\det \nabla \mathbf{y}(\mathbf{x}))
  \mathbf{P}^\mathbf{y}(\mathbf{x}^\mathbf{y})
  (\nabla \mathbf{y}(\mathbf{x}))^{-\top},
  \quad
  \mathbf{x}^\mathbf{y} = \mathbf{y}(\mathbf{x}).
\end{align*}
Using the~relation \eqref{eq:f=fy}, the~tensor can~be~expressed
in~terms of~the~pairwise force function $\mathbf{f}$ by
\begin{align*}
  \mathbf{P}
  &
  (\mathbf{x})
  =
  \frac{\det \nabla \mathbf{y}(\mathbf{x})}{2}
  \int\limits_{ \mathbb{S}^2 }
  \int\limits_{0}^{+\infty} \int\limits_{0}^{+\infty}
    \chi_{\mathbf{y}(\Omega)}(\mathbf{y}(\mathbf{x}) + \alpha \mathbf{m})
    \chi_{\mathbf{y}(\Omega)}(\mathbf{y}(\mathbf{x}) - \beta  \mathbf{m}) \\
    &
    \frac{(\alpha + \beta)^2 \,
          \mathbf{f}(\mathbf{y}^{-1}( \mathbf{y}(\mathbf{x}) + \alpha \mathbf{m} ),
                     \mathbf{y}^{-1}( \mathbf{y}(\mathbf{x}) - \beta  \mathbf{m} )
                    )}
         {(\det \nabla \mathbf{y}(\mathbf{y}^{-1}( \mathbf{y}(\mathbf{x}) + \alpha \mathbf{m} )))
          (\det \nabla \mathbf{y}(\mathbf{y}^{-1}( \mathbf{y}(\mathbf{x}) - \beta  \mathbf{m} )))}
    \otimes (\nabla \mathbf{y}(\mathbf{x}))^{-1} \mathbf{m} \\
  &
  \quad \quad \quad \quad \quad \quad \quad \quad \quad \quad \quad \quad \quad \quad \quad \quad \quad \quad \quad \quad \quad \quad \quad \quad \quad \quad
  \dx \alpha \, \dx \beta \, \dx S(\mathbf{m}),
\end{align*}
where $\chi_{\mathbf{y}(\Omega)}$ is the~characteristic function of~$\Omega$.
The~integrand is to~be understood as~zero whenever the~preimage
of~$\mathbf{y}(\mathbf{x}) + \alpha \mathbf{m}$ or~$\mathbf{y}(\mathbf{x}) - \beta \mathbf{m}$
is not well defined.
This final formula is much more complicated than the~one for~peridynamic tensor $\nu$,
but it~reflects the~fact that although the~constitutive theory in~peridynamics is~primarily being done
in~the~reference configuration,
the~forces exert in~the~deformed one.
This will be~treated in~a~greater detail in~next section
(see Example \ref{ex:nuNPy} and~the~subsequent discussion).

\section{Properties of~the~peridynamic tensors} \label{sec:properties}
In~this section we shall investigate the~symmetry of~$\mathbf{P}^\mathbf{y}$,
compute the~divergence of~$\mathbf{P}$,
and provide an~example of~the~deformation
for~which the~tensors $\mathbf{P}$ and $\nu$ differ.

The~question of~symmetry of~the~peridynamic tensor $\mathbf{P}^\mathbf{y}$
is~quite straightforward in~bond-based peridynamics
(the~earlier version of~the~theory proposed in~\cite{Silling2000}).
Here the~balance of~angular momentum reduces to~the~requirement of~parallelism
of~the~exerting force
\begin{align*}
  \mathbf{f}^\mathbf{y}\xxy \,
  \parallel \,
  (\mathbf{x'}^\mathbf{y} - \mathbf{x}^\mathbf{y})
\end{align*}
and~it can be~seen easily in~\eqref{eq:Py} that the~peridynamic tensor
$\mathbf{P}^\mathbf{y}$ is~indeed symmetric.
In~state-based peridynamics (the~latest version appearing in~\cite{Silling2007}),
however, this~does not seem to~be such an~easy task and the~question still remains open.
We will nevertheless explain in~the~Section \ref{sec:flux} why the~possible asymmetry of~$\mathbf{P}^\mathbf{y}$
actually does not have to~contradict the~balance of~angular momentum,
as it might seem at~first sight.
Yet we consider this partial result about symmetry of~$\mathbf{P}^\mathbf{y}$
interesting and will take a~profit from it.

Next we proceed with investigating the~tensors' divergence.
The~divergence of~$\mathbf{P}$ is~computed using the~knowledge of~the~divergence of~$\mathbf{P}^\mathbf{y}$
and~the~properties of~the~Piola transform
which implies
\begin{align} \label{eq:divPiola}
  \dive \mathbf{P} (\mathbf{x})
  =
  (\det \nabla \mathbf{y}(\mathbf{x})) \dive \mathbf{P}^\mathbf{y}(\mathbf{x}^\mathbf{y}),
  \quad
  \mathbf{x}^\mathbf{y} = \mathbf{y}(\mathbf{x}).
\end{align}
The~divergence of~$\mathbf{P}^\mathbf{y}$ can~be~obtained using~the~theorem
about the~divergence of~the~peridynamic tensor $\nu$,
since these two~tensors are~formally identical, see \eqref{eq:nu} and~\eqref{eq:Py}.
The~only difference (but a~crucial one for~the~mechanical interpretation as it can be seen from Example \ref{ex:nuNPy})
is that the~integration in~the~former is~done in~the~deformed configuration
whereas in~the~later in~the~reference one.
Denoting
\begin{align*}
  \mathcal{I} = \{ (\mathbf{x'}^\mathbf{y}, \mathbf{x}^\mathbf{y}) \in \mathbb{R}^3 \times \mathbb{R}^3: \mathbf{x'}^\mathbf{y} = \mathbf{x}^\mathbf{y} \},
\end{align*}
the~transcription of~the~mentioned theorem reads (c.f. \citep[Theorem 6]{SillingFlux2008}):
\begin{theorem}
  Let a~deformation $\mathbf{y}:\Omega \rightarrow \mathbb{R}^3$ be~given,
  let $\mathbf{f}^\mathbf{y}$ be~the~corresponding pairwise force density,
  and~let~$\mathbf{P}^\mathbf{y}$ be~given by~\eqref{eq:Py}.
  If $\mathbf{f}^\mathbf{y}$ is continuously differentiable
  on~$(\mathbb{R}^3 \setminus \partial \mathbf{y}(\Omega)) \times (\mathbb{R}^3 \setminus \partial \mathbf{y}(\Omega)) \setminus \mathcal{I}$
  and if
  \begin{align*}
    \mathbf{f}^\mathbf{y}(\mathbf{x'}^\mathbf{y}, \mathbf{x}^\mathbf{y})
    =
    o(|\mathbf{x'}^\mathbf{y} - \mathbf{x}^\mathbf{y}|^{-2})
    \quad
    \text{as } |\mathbf{x'}^\mathbf{y} - \mathbf{x}^\mathbf{y}| \rightarrow +\infty,
  \end{align*}
  then
  \begin{align*}
    \dive \mathbf{P}^\mathbf{y} (\mathbf{x}^\mathbf{y})
    =
    \int_{\mathbf{y}(\Omega)} \mathbf{f}^\mathbf{y} \xxy \, \dx \mathbf{x'}^\mathbf{y},
    \quad
    \forall \mathbf{x}^\mathbf{y} \in \mathbb{R}^3.
  \end{align*}
\end{theorem}
The~continuity of~$\mathbf{f}^\mathbf{y}$ in $\mathbf{y}(\Omega)$
is determined both by~the~regularity of~the~deformation
and~by~the~smoothness of~the~constitutive relation
(this means for~instance that there cannot be~a~jump on~the~horizon neither a~blow-up near $\mathcal{I}$);
however, due to ~\eqref{eq:fy=0},
there may be a~discontinuity located
on~$\partial \mathbf{y}(\Omega) \times \partial \mathbf{y}(\Omega)$
which is~therefore excluded.
The~condition of~the~decay at~infinity is~satisfied for~any material with~finite horizon.

This theorem together with~\eqref{eq:divPiola}, \eqref{eq:f=fy}
and substitution formula implies the~desired result
%
%
\begin{multline} \label{eq:divP}
  \dive \mathbf{P} (\mathbf{x})
  =
  (\det \nabla \mathbf{y}(\mathbf{x})) \dive \mathbf{P}^\mathbf{y}(\mathbf{x}^\mathbf{y})
  =
  (\det \nabla \mathbf{y}(\mathbf{x}))
  \int_{\mathbf{y}(\Omega)} \mathbf{f}^\mathbf{y} \xxy \, \dx \mathbf{x'}^\mathbf{y} \\
  =
  \int_{\Omega}
    \mathbf{f}^\mathbf{y} \xxy
    (\det \nabla \mathbf{y}(\mathbf{x'})) (\det \nabla \mathbf{y}(\mathbf{x})) \,
  \dx \mathbf{x'}
  =
  \int_{\Omega}
    \mathbf{f}\xx \,
  \dx \mathbf{x'}.
\end{multline}

Although the~tensors $\mathbf{P}$ and~$\nu$ have the~same divergence,
they are not~eq\-ual, as~the~following example shows.
\begin{example} \label{ex:nuNPy}
  Let us consider a~non-homogeneous deformation $\mathbf{y}:\mathbb{R}^3 \rightarrow \mathbb{R}^3$ given by~the~formula
  \begin{align*}
    \mathbf{y}(x_1, x_2, x_3) = (x_1, x_2 + x_3^3, x_3)
  \end{align*}
  and~a~pairwise force function
  \begin{align*}
    \mathbf{f}\xx
    =
    \gamma(|\mathbf{x'} - \mathbf{x}|)
    (|\mathbf{y}(\mathbf{x'}) - \mathbf{y}(\mathbf{x})|^2 - |\mathbf{x'} - \mathbf{x}|^2)
    (\mathbf{y}(\mathbf{x'}) - \mathbf{y}(\mathbf{x}))
  \end{align*}
  which is~a~particular example of~the~class of~materials introduced in~\citep[eq. (49)]{Silling2000}.
  The~so-called \emph{shielding function} $\gamma:[0, +\infty) \rightarrow [0, +\infty)$
  is~supposed to~be~sufficiently smooth and~to~vanish for~$|\mathbf{x'} - \mathbf{x}| \geq \delta$,
  where $\delta$ is~the~horizon.
  Since
  \begin{align} \label{eq:ExGrad0}
    \nabla \mathbf{y}(\mathbf{0}) = \mathbf{I},
  \end{align}
  the~Piola transform is~also identity and hence $\mathbf{P}(\mathbf{0}) = \mathbf{P}^\mathbf{y}(\mathbf{0})$.
  This implies that $\mathbf{P}(\mathbf{0})$ is~symmetric
  and~therefore it is sufficient to~show the~non-symmetry of~$\nu(0)$.
  To~show this we compute $\nu_{12}(0)$ and~$\nu_{21}(0)$
  which will turn out to~be~different.
  The~integral over $\mathbb{S}^2$ in~\eqref{eq:nu} can be~rewritten using spherical coordinates
  \begin{align*}
    \mathbf{m} &= (\cos \theta, \cos \varphi \sin \theta, \sin \varphi \sin \theta),
    \quad
    \dx S(\mathbf{m}) = \sin \theta \, \dx \theta \dx \varphi, \\
    \theta &\in (0, \pi),
    \quad
    \varphi \in (0, 2\pi)
  \end{align*}
  and then a~straightforward calculation shows that
  \begin{align*}
    \nu_{12}(0)
    &=
    \frac{4\pi}{35}
    \int\limits_0^{+\infty} \int\limits_0^{+\infty}
      (\alpha + \beta)^4 (\alpha^3 + \beta^3) \gamma(\alpha + \beta) \,
    \dx \alpha \dx \beta
  \end{align*}
  and
  \begin{align*}
    \nu_{21}(0)
    &=
    \nu_{12}(0)
    +
    \frac{2\pi}{11}
    \int\limits_0^{+\infty} \int\limits_0^{+\infty}
      (\alpha + \beta)^2 (\alpha^3 + \beta^3)^3 \gamma(\alpha + \beta) \,
    \dx \alpha \dx \beta.
  \end{align*}
  Hence, for~an~appropriate choice of~the~function $\gamma$, indeed $\nu_{12}(0) \neq \nu_{21}(0)$.
\end{example}

The~reason why the~tensor
\begin{multline*}
  \nu^\mathbf{y}(\mathbf{x})
  =
  \frac{1}{2 \det \nabla \mathbf{y}(\mathbf{x})} \\
    \int\limits_{\mathbb{S}^{2}} \int\limits_0^{+\infty} \int\limits_0^{+\infty}
      (\alpha + \beta)^2 \mathbf{f}(\mathbf{x} + \alpha\mathbf{m}, \mathbf{x} - \beta\mathbf{m}, t)
      \otimes
      \nabla \mathbf{y}(\mathbf{x}) \mathbf{m} \,
    \dx \alpha \dx \beta \dx S(\mathbf{m})
\end{multline*}
(the~inverse Piola transform of~$\nu(\mathbf{x})$)
is not symmetric at~the~origin
,while the~tensor $\mathbf{P}^\mathbf{y}$ is,
lies in~the~fact that the~original tensor $\nu$ is computed in~the~reference configuration,
taking into account only local transformation of~geometry described by $\nabla \mathbf{y}$.
Since peridynamics is~a~non-local theory
the~transformation within the~whole horizon has to~be~incorporated
(which is done in~the~definition of~$\mathbf{P}^{\mathbf{y}}$).
For~a~non-homogeneous deformation, such as the~one in~the~example,
it happens for~$\mathbf{x'} = \alpha \mathbf{m}$ and~$\mathbf{x} = -\beta \mathbf{m}$ that
\begin{align*}
  \mathbf{x'} - \mathbf{x} \parallel \mathbf{m}
  \quad
  \text{, whilst}
  \quad
  \mathbf{f}\xx \nparallel \nabla \mathbf{y}(\mathbf{0}) \mathbf{m} = \mathbf{m}
\end{align*}
for almost every $\mathbf{m} \in \mathbb{S}^2$ (see Fig. \ref{fig:Example} for~illustration).
This means that $\mathbf{m}$ is not~the~direction under which the~points $\mathbf{x'}$ and~$\mathbf{x}$
exert force upon each other in~the~deformed configuration
and hence $\nu(0)$ does not describe the~force-flux properly.
This geometric inconsistency makes therefore the~mechanical interpretation of~$\nu$
presented in~\citep[sec. 6]{SillingFlux2008} problematic.

\begin{figure}
  \includegraphics[width=\textwidth]{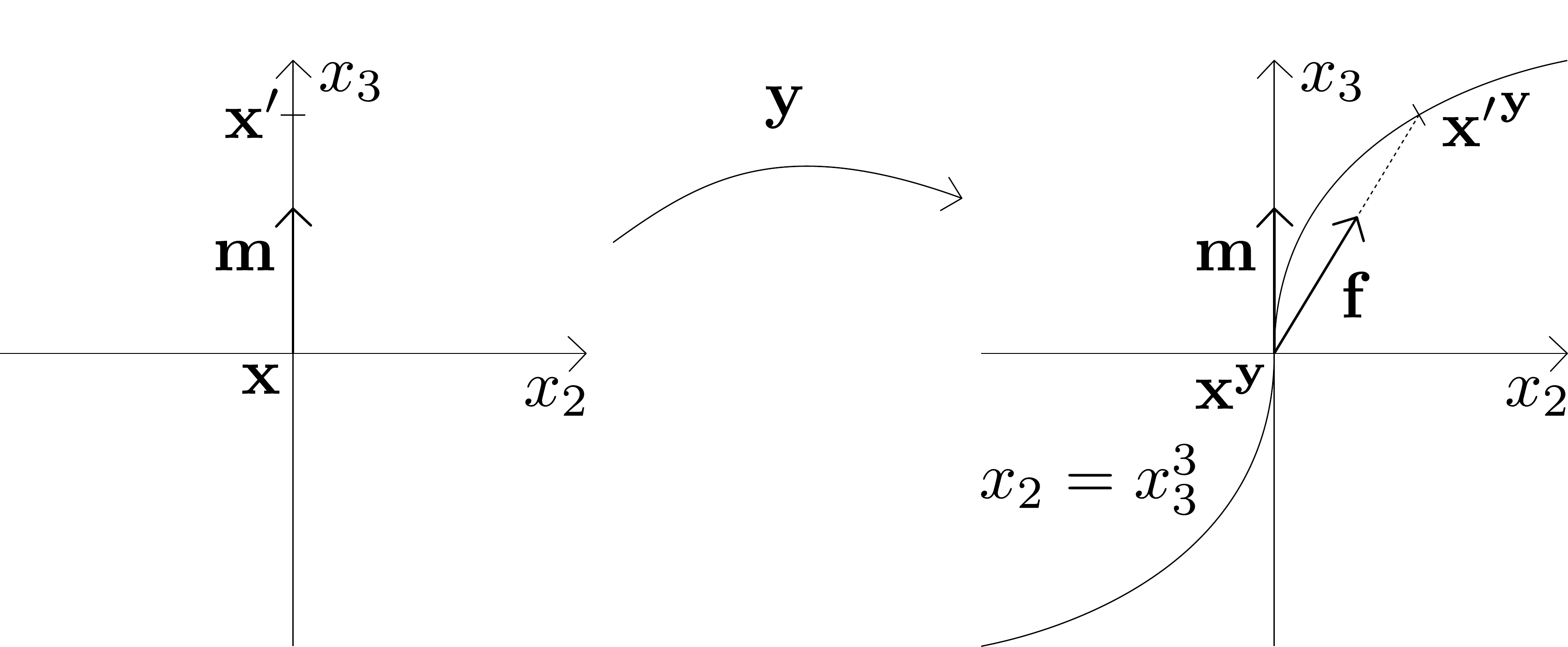}
  \caption{The~geometric inconsistency in~the~definition of~the~tensor $\nu$}
  \label{fig:Example}
\end{figure}

We end up this section by~an~investigation of~the~limiting behavior
of~the~peridynamic tensor $\mathbf{P}$ for~horizon tending to~zero.
To~do so,
we have to~specify first what~the~pairwise force function $\mathbf{f}$ depends on.
For~the~sake of~clarity, we restrict ourselves to~a~simpler constitutive relation of~bond-based peridynamics
thought the~following steps may be easily generalized also for~state-based peridynamics.
Let then
\begin{align*}
  \mathbf{f} \xx
  =
  \hat{\mathbf{f}}(\mathbf{y}(\mathbf{x'}) - \mathbf{y}(\mathbf{x}), \mathbf{x'} - \mathbf{x}, \mathbf{x}),
\end{align*}
where $\hat{\mathbf{f}}: \mathbb{R}^3 \times \mathbb{R}^3 \times \Omega \rightarrow \mathbb{R}^3$ is such that
\begin{align*}
  \hat{\mathbf{f}}(\cdot, \tilde{\mathbf{x}}, \cdot) \equiv 0
  \quad
  \text{whenever } |\tilde{\mathbf{x}}| \geq \delta
\end{align*}
and $\delta$ is some fixed horizon.
After performing the~scaling in~the~same spirit as~in~\cite{SillingLehoucq2008}
we obtain a~family of~peridynamic tensors
\begin{multline} \label{eq:Ps}
  \mathbf{P}_s
  (\mathbf{x})
  :=
  \frac{\det \nabla \mathbf{y}(\mathbf{x})}{2}
  \int\limits_{ \mathbb{S}^2 }
  \int\limits_{0}^{+\infty} \int\limits_{0}^{+\infty} \\
    (\alpha + \beta)^2
    \chi_{\mathbf{y}(\Omega)}(\mathbf{y}(\mathbf{x}) + s\alpha \mathbf{m})
    \chi_{\mathbf{y}(\Omega)}(\mathbf{y}(\mathbf{x}) - s\beta  \mathbf{m}) \\
    \frac{\hat{\mathbf{f}}((\alpha + \beta) \mathbf{m},
                           \frac{\mathbf{y}^{-1}( \mathbf{y}(\mathbf{x}) + s\alpha \mathbf{m} )
                                 -
                                 \mathbf{y}^{-1}( \mathbf{y}(\mathbf{x}) - s\beta  \mathbf{m} )}
                                {s},
                           \mathbf{y}^{-1}( \mathbf{y}(\mathbf{x}) - s\beta  \mathbf{m} )
                          )}
         {(\det \nabla \mathbf{y}(\mathbf{y}^{-1}( \mathbf{y}(\mathbf{x}) + s\alpha \mathbf{m} )))
          (\det \nabla \mathbf{y}(\mathbf{y}^{-1}( \mathbf{y}(\mathbf{x}) - s\beta  \mathbf{m} )))} \\
  \quad \quad \quad \quad \quad \quad \quad \quad \quad \quad \quad \quad \quad \quad \quad \quad \quad \quad \quad
    \otimes (\nabla \mathbf{y}(\mathbf{x}))^{-1} \mathbf{m} \,
  \dx \alpha \, \dx \beta \, \dx S(\mathbf{m})
\end{multline}
indexed by~the~dimensionless parameter $s \searrow  0_+$
which measures the~non-locality.

Although the~formulae for~the~tensors $\mathbf{P}$ and $\nu$ seem to~be very different at~first sight,
the~following theorem shows that their collapsed counterparts are surprisingly equal.
\begin{theorem}
  Let~$\Omega \subset \mathbb{R}^3$ be~bounded domain and~$\overline{\Omega}$ denote its closure.
  Let moreover $\hat{\mathbf{f}}: \mathbb{R}^3 \times \mathbb{R}^3 \times \overline{\Omega} \rightarrow \mathbb{R}^3$
  be~continuous,
  $\mathbf{y}: \Omega \rightarrow \mathbf{y}(\Omega)$
  and $\mathbf{y}^{-1}: \mathbf{y}(\Omega) \rightarrow \Omega$ continuously differentiable
  and $\det \nabla \mathbf{y} > 0$ in~$\Omega$.
  Then
  \begin{align*}
    \mathbf{P}_0(\mathbf{x})
    :=
    \lim_{s \rightarrow 0_+} \mathbf{P}_s
    =
    \nu_0(\mathbf{x}).
  \end{align*}
\end{theorem}
\begin{proof}
  It is easy to see that the~integrand in~\eqref{eq:Ps} converges point-wisely to
  \begin{align*}
    \left( \frac{\alpha + \beta}{\det \nabla \mathbf{y}(\mathbf{x})} \right)^2
    \hat{\mathbf{f}}((\alpha + \beta) \mathbf{m}, (\nabla \mathbf{y}(\mathbf{x}))^{-1} (\alpha + \beta) \mathbf{m}, \mathbf{x} )
    \otimes
    (\nabla \mathbf{y}(\mathbf{x}))^{-1} \mathbf{m}.
  \end{align*}
  Thanks to our assumptions we may use the~Lebesgue dominated convergence theorem
  to~interchange the~order of~limit and~integration
  and hence
  \begin{multline*}
    \mathbf{P}_0(\mathbf{x})
    =
    \frac{1}{2 \det \nabla \mathbf{y}(\mathbf{x})}
    \int\limits_{ \mathbb{S}^2 }
    \int\limits_{0}^{+\infty} \int\limits_{0}^{+\infty}
      (\alpha + \beta)^2
      \\
      \hat{\mathbf{f}}((\alpha + \beta) \mathbf{m}, (\nabla \mathbf{y}(\mathbf{x}))^{-1} (\alpha + \beta) \mathbf{m}, \mathbf{x} )
      \otimes
      (\nabla \mathbf{y}(\mathbf{x}))^{-1} \mathbf{m} \,
    \dx \alpha \, \dx \beta \, \dx S(\mathbf{m}).
  \end{multline*}

  Similarly as~in~\cite{SillingLehoucq2008},
  we perform the~substitution
  \begin{align*}
    \alpha = p - \beta, \quad \dx \alpha = \dx p, \quad p \in (\beta, +\infty)
  \end{align*}
  leading to
  \begin{multline*}
    \mathbf{P}_0(\mathbf{x})
    :=
    \frac{1}{2 \det \nabla \mathbf{y}(\mathbf{x})}
    \int\limits_{ \mathbb{S}^2 }
    \int\limits_{0}^{+\infty} \int\limits_{\beta}^{+\infty} \\
      p^2\,
      \hat{\mathbf{f}}(p \mathbf{m}, (\nabla \mathbf{y}(\mathbf{x}))^{-1} p \mathbf{m}, \mathbf{x} )
      \otimes
      (\nabla \mathbf{y}(\mathbf{x}))^{-1} \mathbf{m} \,
    \dx p \, \dx \beta \, \dx S(\mathbf{m}).
  \end{multline*}
  Using the~Fubini theorem, we may interchange the~order of~integration with respect to~$p$ and~$\mathbf{\beta}$
  and then integrate $\beta$ form $0$ to~$p$
  which yields an~additional power of~$p$.
  Hence
  \begin{multline*}
    \mathbf{P}_0(\mathbf{x})
    =
    \frac{1}{2 \det \nabla \mathbf{y}(\mathbf{x})}
    \int\limits_{ \mathbb{S}^2 }
    \int\limits_{0}^{+\infty}
      p^2\,
      \hat{\mathbf{f}}(p \mathbf{m}, (\nabla \mathbf{y}(\mathbf{x}))^{-1} p \mathbf{m}, \mathbf{x} )
      \otimes
      (\nabla \mathbf{y}(\mathbf{x}))^{-1} p \mathbf{m} \, \\
    \dx p \, \dx S(\mathbf{m}),
  \end{multline*}
  which is~nothing but~the~volume integral over $\mathbb{R}^3$ with respect to~$\tilde{\mathbf{y}} := p \mathbf{m}$,
  i.e.
  \begin{align*}
    \mathbf{P}_0(\mathbf{x})
    =
    \frac{1}{2 \det \nabla \mathbf{y}(\mathbf{x})}
    \int\limits_{ \mathbb{R}^3 }
      \hat{\mathbf{f}}(\tilde{\mathbf{y}}, (\nabla \mathbf{y}(\mathbf{x}))^{-1} \tilde{\mathbf{y}}, \mathbf{x} )
      \otimes
      (\nabla \mathbf{y}(\mathbf{x}))^{-1} \tilde{\mathbf{y}} \,
    \dx \tilde{\mathbf{y}}.
  \end{align*}
  Finally using a~substitution
  \begin{align*}
    \tilde{\mathbf{y}} = \nabla \mathbf{y}(\mathbf{x}) \, \tilde{\mathbf{x}},
    \quad
    \tilde{\mathbf{x}} \in \mathbb{R}^3
  \end{align*}
  we see that
  \begin{align*}
    \mathbf{P}_0(\mathbf{x})
    =
    \frac{1}{2}
    \int\limits_{ \mathbb{R}^3 }
      \hat{\mathbf{f}}(\nabla \mathbf{y}(\mathbf{x}) \tilde{\mathbf{x}}, \tilde{\mathbf{x}}, \mathbf{x} )
      \otimes
      \tilde{\mathbf{x}} \,
    \dx \tilde{\mathbf{x}}
  \end{align*}
  which is exactly $\nu_0(\mathbf{x})$ for~bond-based peridynamics (see \citep[eq. (50)]{SillingLehoucq2008}).
\end{proof}


This result seems to~be of~a~particular interest.
It shows that the~geometric inconsistency contained in~the~peridynamic tensor $\nu$
vanishes in~the~li\-mit of~small horizon.
Moreover it implies that identifying of~the~1st Piola-Kirchhoff stress tensor with the~collapsed tensor $\nu_0$
is now equivalent to~identifying it with the~collapsed tensor $\mathbf{P}_0$.
In~the~next section we shall argue that such an~identification is possible.

\section{The~correspondence between the~Cauchy stress tensor and the~peridynamic stress tensor $\mathbf{P}^\mathbf{y}$}
\label{sec:flux}
A~very natural question is whether the~peridynamic tensor $\mathbf{P}^\mathbf{y}$
is indeed the~Cauchy stress tensor $\mathbf{T}^\mathbf{y}$.
We shall prove that it is \emph{not} the~case.
We begin with a~generalization of~the~procedure used in~the~Section \ref{sec:derivation}
for~determining the~form of~the~peridynamic tensor $\mathbf{P}^\mathbf{y}$.
This will lead us to~a~general expression for~force flux between two sufficiently regular regions.
Based on~its knowledge we shall conclude what the~relation between these tensors is.

The~necessary connection between the~local and non-local interaction is provided
by~the~formulae \eqref{eq:PFFunction}, \eqref{eq:f=fy}, and~\eqref{eq:CVector}.
For~a~given material and~its deformation we need to~find a~vector field
$\mathbf{t}^\mathbf{y}:\mathbf{y}(\Omega)\times\mathbb{S}^2 \rightarrow \mathbb{R}^3$
s.t. for~any two adjacent spatial regions $A^\mathbf{y}, B^\mathbf{y} \subset \mathbf{y}(\Omega)$
the~mutual force interaction can be expressed as~the~surface integral
of~$\mathbf{t}^\mathbf{y}$, i.e.
\begin{align} \label{eq:PFFy-CT}
  \int_{B^\mathbf{y}} \int_{A^\mathbf{y}}
      \mathbf{f}^\mathbf{y}\xxy \,
    \dx \mathbf{x'}^\mathbf{y} \dx \mathbf{x}^\mathbf{y}
  = \int_{\partial A^\mathbf{y} \cap \partial B^\mathbf{y}}
      \mathbf{t}^\mathbf{y} \sny \,
    \dx S\sy.
\end{align}
Moreover, in~simple materials
the~vector $\mathbf{t}^\mathbf{y}$ can depend on~the~surface
$\partial A^\mathbf{y} \cap \partial B^\mathbf{y}$
only through its normal vector $\mathbf{n}^\mathbf{y}$ at~point $\mathbf{s}^\mathbf{y}$.

The~problem is that the~substitution \eqref{eq:CauchySubst},
thought it seems to~be very natural,
cannot be~applied in~general.
For~example there may be several intersections of~the~line segment $[\mathbf{x'}^\mathbf{y}, \mathbf{x}^\mathbf{y}]$
and~$\partial A^\mathbf{y} \cap \partial B^\mathbf{y}$
or~there may be even no~intersection at~all (see Fig.~\ref{fig:FluxBad}).
Nevertheless, in~the~case that there is exactly one intersection for~every interacting pair
(which holds for~example for~a~convex set and~its complement),
we may proceed further.
The~difference is that now $\mathbf{s}^\mathbf{y}$ is a~map which locally describes the~boundary
and the~range of~the~lengths $\alpha$ and~$\beta$
may depend on~the~direction $\mathbf{m}$ 
(see Fig.~\ref{fig:FluxBad} for~illustration).
Otherwise the~calculations are performed in~a~similar way
yielding the~following form of~the~mutual force
\begin{multline*}
  F\ABy \\
  =
  \int\limits_{\mathcal{S}}
  \int\limits_{ \mathbb{S}^2_+(\mathbf{s}^\mathbf{y}) }
  \int\limits_{0}^{\tilde{\beta}} \int\limits_{0}^{\tilde{\alpha}}
    (\alpha + \beta)^2 \,
    \mathbf{f}^\mathbf{y}(\mathbf{s}^\mathbf{y} + \alpha \mathbf{m},
                          \mathbf{s}^\mathbf{y} - \beta \mathbf{m})
    (\mathbf{m} \cdot \mathbf{n}^\mathbf{y}) \,
  \dx \alpha \, \dx \beta \, \dx S(\mathbf{m}) \, \dx S(\mathbf{s}^\mathbf{y}),
\end{multline*}
where
\begin{align*}
  \mathcal{S}
  =
  \partial A^\mathbf{y} \cap \partial B^\mathbf{y},
  \quad
  \mathbb{S}^2_+(\mathbf{s}^\mathbf{y})
  =
  \{ \mathbf{m} \in \mathbb{S}^{2}: \mathbf{m}\cdot\mathbf{n}^\mathbf{y}(\mathbf{s}^\mathbf{y}) > 0 \}
\end{align*}
and $\tilde{\alpha}$ and~$\tilde{\beta}$ are~functions
\begin{align*}
  \tilde{\alpha} = \tilde{\alpha}(A^\mathbf{y}, \mathbf{s}^\mathbf{y}, \mathbf{m}),
  \quad
  \tilde{\beta} = \tilde{\alpha}(B^\mathbf{y}, \mathbf{s}^\mathbf{y}, \mathbf{m})
\end{align*}
such that
\begin{align*}
  \forall \mathbf{s}^\mathbf{y} \in \mathcal{S} \,
  \forall \mathbf{m} \in \mathbb{S}^2_+(\mathbf{s}^\mathbf{y}):
  \quad
  \mathbf{s}^\mathbf{y} + \alpha \mathbf{m} \in A^\mathbf{y}
  &\Leftrightarrow
  \alpha \in (0, \tilde{\alpha})
  \text{ and} \\
  \mathbf{s}^\mathbf{y} - \beta \mathbf{m} \in B^\mathbf{y}
  &\Leftrightarrow
  \beta \in (0, \tilde{\beta}).
\end{align*}

\begin{figure}
  \includegraphics[width=\textwidth]{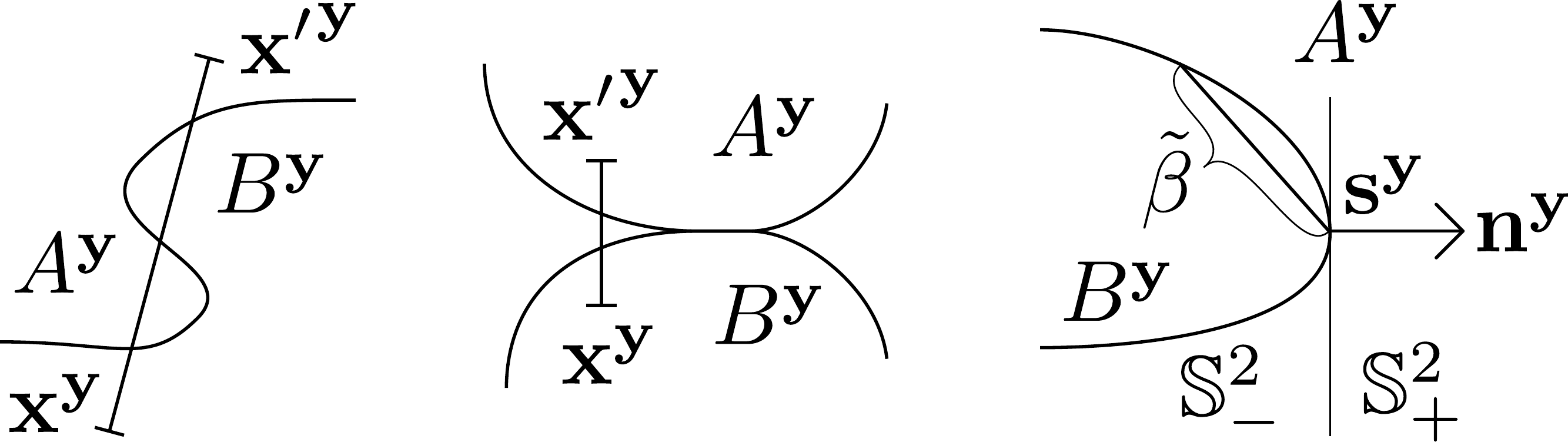}
  \caption{Issues arising from~non-locality}
  \label{fig:FluxBad}
\end{figure}

Hence the~force flux
\begin{align} \label{eq:tyNonlocal}
  \mathbf{t}^\mathbf{y} \sy
  = 
  \int\limits_{ \mathbb{S}^2_+(\mathbf{x}^\mathbf{y}) }
  \int\limits_{0}^{\tilde{\beta}} \int\limits_{0}^{\tilde{\alpha}}
    (\alpha + \beta)^2 \,
    \mathbf{f}^\mathbf{y}(\mathbf{s}^\mathbf{y} + \alpha \mathbf{m},
                          \mathbf{s}^\mathbf{y} - \beta \mathbf{m})
    (\mathbf{m} \cdot \mathbf{n}^\mathbf{y}) \,
  \dx \alpha \, \dx \beta \, \dx S(\mathbf{m})
\end{align}
from $A^\mathbf{y}$ to~$B^\mathbf{y}$ is~strictly non-local
since it~depends on~the~contact surface $\partial A^\mathbf{y} \cap \partial B^\mathbf{y}$
not only through~the~normal vector $\mathbf{n}^\mathbf{y}$ at~a~point,
but it~involves its~nontrivial part close to~the~point $\mathbf{s}^\mathbf{y}$.
To see this, let us consider a~ball $C^\mathbf{y} \subset \mathbf{y}(\Omega)$ and~its tangent plane $\mathcal{P}^\mathbf{y}$.
Let their intersection be denoted as $\mathbf{s}^\mathbf{y}$.
Moreover let $A^\mathbf{y}$ be~the~half-space containing the~ball $C^\mathbf{y}$
and $B^\mathbf{y}$ denote the~other one (see Fig.~\ref{fig:FluxDiff}).
Although the~normal vector $\mathbf{n}^\mathbf{y}$ at~the~point $\mathbf{s}^\mathbf{y}$ is the~same
for~both surfaces,
the~fluxes from $A^\mathbf{y}$ to~$B^\mathbf{y}$
and from $C^\mathbf{y}$ to~its complement differ at~this point.
This fact actually shows that the~peridynamic non-local force interaction
cannot be described by~a~tensor in~the~sense of~\eqref{eq:CVector} and~\eqref{eq:CTensor},
which is, however, a~fundamental assumption in~the~classical theory of~simple materials.

This resembles to~the~situation in~so-called non-simple materials
(which can be, according to~\cite{BazantJirasek2002}, understood in some sense as non-local)
where the~Cauchy stress vector may depend, besides the~surface normal vector $\mathbf{n}^\mathbf{y}$,
also on~the surface curvature (c.f. \cite{Toupin1962}, \cite{Toupin1964} or~\cite{Fried2006}).
It should be also noted that the~possible asymmetry in~the~integration bounds of~$\alpha$ and~$\beta$
makes the~step from~\eqref{eq:tyHemi} to~\eqref{eq:tyWhole} impossible in~general.
It is therefore very surprising that the~tensor $\mathbf{P}^\mathbf{y}$
provides, by~its divergence, the~correct total force flux from $A^\mathbf{y}$ to~$B^\mathbf{y}$
despite the~fact that the~two fluxes \eqref{eq:tyWhole} and \eqref{eq:tyNonlocal}
differ at~each point where the~boundary is~curved.

\begin{figure}
  \centering
  \includegraphics[width=0.4\textwidth]{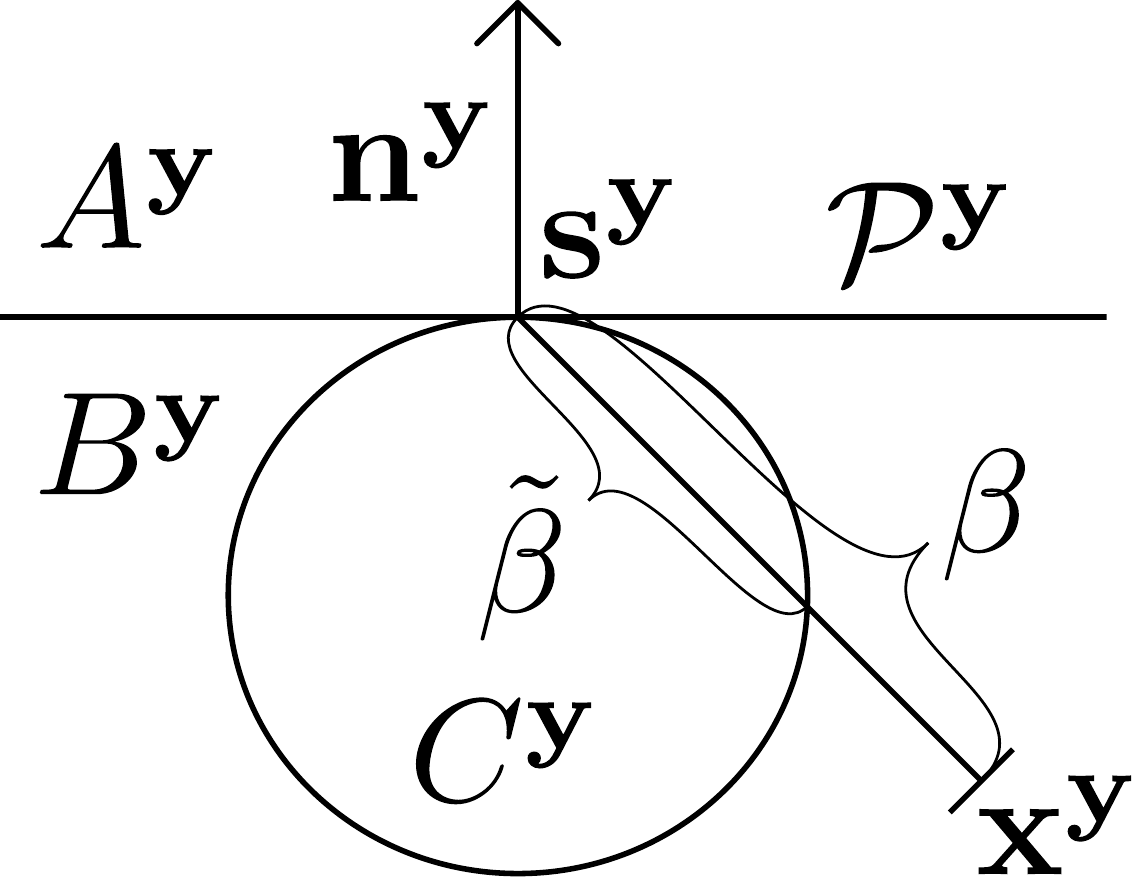}
  \caption{Different fluxes for~the~same unit normal vector.}
  \label{fig:FluxDiff}
\end{figure}

Yet we anticipate that in~the~limit of~vanishing non-locality the~situation changes
and that the~collapsed peridynamic tensor $\mathbf{P}^{\mathbf{y}}_0$
can be identified with the~Cauchy stress tensor (and hence $\mathbf{P}_0$ with the~1st Piola-Kirchhoff tensor).
The~idea is as~follows.
Let~two adjacent regions
$\mathbf{A}^\mathbf{y},\, \mathbf{B}^\mathbf{y} \subset \mathbf{y}(\Omega)$
be such that the~interacting pairs $(\mathbf{x'}^\mathbf{y},\, \mathbf{x}^\mathbf{y})$
for~which the~substitution \eqref{eq:CauchySubst} cannot be used
vanish in~the~limit.
If moreover the~common boundary at~some point $\mathbf{s}^\mathbf{y} \in \partial A^\mathbf{y} \cap \partial B^\mathbf{y}$
can be~approximated by~its tangent plane,
then the~force flux \eqref{eq:tyNonlocal} through the~boundary at~the~point $\mathbf{s}^\mathbf{y}$ is likely to~converge
to~the~force flux \eqref{eq:tyWhole} through its tangent plane.
Hence it~seems that in~the~limit the~equation \eqref{eq:PFFy-CT}
may hold for~broader class of~adjacent regions
than just two parts of~the~deformed body divided by a~plane.

\section*{Discussion}
In~our opinion, the~biggest advantage of~our approach
is the~use of~the substitution formula \eqref{eq:CauchySubst}
which provides a~better insight into the~problem.
It~leads to~an~expression \eqref{eq:tyNonlocal} for~the~non-local force flux
which is subsequently helpful for~proving
that the~peridynamic interaction cannot be described by a~tensor in~general;
however, when only interactions through planes are considered,
the~formula can be simplified to~\eqref{eq:tyWhole} and~the~flux has a~tensorial character.

Based upon this simplified formula we defined the~peridynamic tensor $\mathbf{P}$,
whose divergence turned out to~be of~the~correct form \eqref{eq:divP}.
Moreover, thanks to~this mechanical interpretation
(i.e. a~force flux through a~plane),
the~derived peridynamic tensor $\mathbf{P}$
seems to~be more convenient than the~peridynamic tensor $\nu$
which has only the~correct divergence.
The~problem with the~tensor $\nu$, as~shown in~the~Example \ref{ex:nuNPy},
is that it is computed in~the reference configuration
disregarding the~non-local transformation of~the~geometry due to~the~deformation.
This example also shows that the~two tensors $\nu$ and~$\mathbf{P}$ differ
though some kind of~uniqueness result for~the~former was presented in~\cite{SillingFlux2008}.

As~was already mentioned, the~formula \eqref{eq:tyNonlocal} for~the~non-local force flux shows
that the~peridynamic tensor $\mathbf{P}$ can not be considered as~the~1st Piola-Kirchhoff stress tensor.
On~the~other hand, the~same formula may be used for~proving
that in~the~limit of~vanishing non-locality this is no longer true for~the~collapsed tensor $\mathbf{P}_0$.
Fortunately it also holds that this tensor coincides with the~collapsed peridynamic tensor $\nu_0$
computed in~\cite{SillingLehoucq2008}.
This provides the~explanation
why the~tensor $\nu_0$ may provide the~description of~the~limiting model in~local elasticity.

The~last thing to~be discussed is symmetry of~the~tensors.
In~the~bond-based peridynamic there is no ambiguity left
since the~peridynamic tensors $\mathbf{P}^\mathbf{y}$ and $\mathbf{P}$ pose the~right symmetries.
In~the~state based peridynamic, the~question of~symmetry remains open;
however, since these tensors are no longer supposed to~coincide with the~Cauchy and the~1st Piola-Kirchhoff tensor respectively,
their symmetry is no longer relevant.
What is of true importance is the~symmetry of~their collapsed counterparts
$\mathbf{P}^\mathbf{y}_0$ and $\mathbf{P}_0$.
Since it holds that the~tensors $\mathbf{P}_0$ and $\nu_0$ are equal,
one can use the~result for~the~latter which was proved in~\cite{SillingLehoucq2008}.
It says that the~tensor $\nu_0$ poses the~same symmetry as the~1st Piola-Kirchhoff does
provided the~balance of~angular momentum in~state-based peridynamic is satisfied.

Concerning the~further research,
the~possible next step could be the~incorporation of~the~boundary conditions
both for~finite horizon and~the~limiting case.
It~would be also worthy to~specify the~sufficient regularity
under which the~identification of~the~collapsed tensor $\mathbf{P}_0$
and the~1st Piola-Kirchhoff tensor may be proved rigorously.
The~main difficulty is to~select a~family of~sufficiently regular control volumes
which is at~the~same time preserved by~the~deformation.
Despite a~lot of~effort (c.f. \cite{Noll1974} or \cite{Ziemer1983})
no~such a~selection is still perfect (see \cite{Noll2010}).

\section*{Acknowledgments}
This research was performed within the grants
MŠMT project 7AMB16AT015,
GAČR-FWF project 16-34894L,
DAAD-AVČR project DAAD-16-14
and SVV-2017-260455.
The author is deeply thankful to~Martin Kružík, Ondřej Souček
and Vít Průša for~inspiring conceptual discussions.

\bibliographystyle{plainnat}     
\bibliography{literatura}

\begin{thebibliography}{18}
\providecommand{\natexlab}[1]{#1}
\providecommand{\url}[1]{\texttt{#1}}
\expandafter\ifx\csname urlstyle\endcsname\relax
  \providecommand{\doi}[1]{doi: #1}\else
  \providecommand{\doi}{doi: \begingroup \urlstyle{rm}\Url}\fi

\bibitem[Bažant and Jirásek(2002)]{BazantJirasek2002}
Z.P. Bažant and M.~Jirásek.
\newblock Nonlocal integral formulations of plasticity and damage: Survey of
  progress.
\newblock \emph{Journal of Engineering Mechanics}, 128\penalty0 (11):\penalty0
  1643--1670, 2002.

\bibitem[Bobaru and Hu(2012)]{BobaruHu2012}
F.~Bobaru and W.~Hu.
\newblock The meaning, selection, and use of the peridynamic horizon and its
  relation to crack branching in brittle materials.
\newblock \emph{International Journal of Fracture}, 176\penalty0 (2):\penalty0
  215--222, 2012.

\bibitem[Cauchy(1828)]{Cauchy1828}
A.~L. Cauchy.
\newblock \emph{De~la~pression ou~tension dans un~système de~points
  matériels}.
\newblock A Paris, chez De Bure frères, Libraires du Roi et de la Bibliotheque
  du Roi, 1828.

\bibitem[Ciarlet(1988)]{Ciarlet1988}
G.P. Ciarlet.
\newblock \emph{Mathematical Elasticity, Volume I: Three-dimensional
  Elasticity}.
\newblock Elsevier Science Publisher, 1988.

\bibitem[Emmrich et~al.(2013)Emmrich, Lehoucq, and Puhst]{EmmrichLehoucq2012}
E.~Emmrich, R.B. Lehoucq, and D.~Puhst.
\newblock \emph{Peridynamics: A Nonlocal Continuum Theory}, pages 45--65.
\newblock Springer Berlin Heidelberg, 2013.

\bibitem[Fried and E.(2006)]{Fried2006}
E.~Fried and Gurtin~M. E.
\newblock Tractions, balances, and~boundary conditions for~nonsimple materials
  with application to~liquid flow at~small-length scales.
\newblock \emph{Arch. Rational Mech. Anal.}, 182:\penalty0 513--554, 2006.

\bibitem[Gurtin et~al.(2010)Gurtin, Fried, and Anand]{GurtinAnand2010}
M.E. Gurtin, E.~Fried, and L.~Anand.
\newblock \emph{The Mechanics and Termodynamics of Continua}.
\newblock Cambridge University Press, 2010.

\bibitem[Love(1892)]{Love1892}
A.~Love.
\newblock \emph{A~Treatise on~the~Mathematical Theory of~Elasticity}.
\newblock Cambridge University Press, 1892.

\bibitem[Noll(1974)]{Noll1974}
W.~Noll.
\newblock \emph{The~Foundations of~Classical Mechanics in~the~Light of~Recent
  Advances in~Continuum Mechanics}, pages 31--47.
\newblock Springer Berlin Heidelberg, Berlin, Heidelberg, 1974.

\bibitem[Noll(2010)]{Noll2010}
W.~Noll.
\newblock Thoughts on the concept of stress.
\newblock \emph{Journal of Elasticity}, 100\penalty0 (1):\penalty0 25--32,
  2010.

\bibitem[Silling(2000)]{Silling2000}
S.A. Silling.
\newblock Reformulation of elasticity theory for discontinuities and long range
  forces.
\newblock \emph{Journal of the Mechanics and Physics of Solids}, 48:\penalty0
  175--209, 2000.

\bibitem[Silling and Lehoucq(2008{\natexlab{a}})]{SillingFlux2008}
S.A. Silling and R.B. Lehoucq.
\newblock Force flux and the peridynamic stress tensor.
\newblock \emph{Journal of the Mechanics and Physics of Solids}, 56\penalty0
  (4):\penalty0 1566–1577, 2008{\natexlab{a}}.

\bibitem[Silling and Lehoucq(2008{\natexlab{b}})]{SillingLehoucq2008}
S.A. Silling and R.B. Lehoucq.
\newblock Convergence of peridynamics to classical elasticity theory.
\newblock \emph{Journal of Elasticity}, 93:\penalty0 13--37,
  2008{\natexlab{b}}.

\bibitem[Silling and Lehoucq(2010)]{SillingLehoucq2010}
S.A. Silling and R.B. Lehoucq.
\newblock Peridynamic theory of solid mechanics.
\newblock \emph{Advances in applied mechanics}, 44:\penalty0 74--168, 2010.

\bibitem[Silling et~al.(2007)Silling, Epton, Weckner, Xu, and
  Askari]{Silling2007}
S.A. Silling, E.~Epton, O.~Weckner, J.~Xu, and E.~Askari.
\newblock Peridynamic states and constitutive modeling.
\newblock \emph{Journal of Elasticity}, 88:\penalty0 151--184, 2007.

\bibitem[Toupin(1962)]{Toupin1962}
R.~A. Toupin.
\newblock Elastic materials with couple-stresses.
\newblock \emph{Archive for~Rational Mechanics and~Analysis}, 11\penalty0
  (1):\penalty0 385--414, 1962.

\bibitem[Toupin(1964)]{Toupin1964}
R.~A. Toupin.
\newblock Theories of~elasticity with couple-stress.
\newblock \emph{Archive for~Rational Mechanics and~Analysis}, 17\penalty0
  (2):\penalty0 85--112, 1964.

\bibitem[Ziemer(1983)]{Ziemer1983}
W.~P. Ziemer.
\newblock Cauchy flux and~sets of~finite perimeter.
\newblock \emph{Archive for Rational Mechanics and Analysis}, 84\penalty0
  (3):\penalty0 189--201, 1983.

\end{thebibliography}

\newpage
\todototoc
\listoftodos[Notes]

\end{document}